\documentclass{fmcad}

\usepackage{graphicx}
\usepackage{amsmath,amsfonts}
\usepackage{listings}
\usepackage{tcolorbox}
\usepackage{xcolor}
\usepackage{algorithm}
\usepackage{algorithmic}
\usepackage{xspace}
\usepackage{hyperref}                   
\usepackage{cleveref}
\usepackage{multirow}
\usepackage{booktabs}
\usepackage{subfig}
\usepackage{stfloats}

\usepackage{amsthm}
\theoremstyle{plain}
\newtheorem{theorem}{Theorem}[section]   
\newtheorem{lemma}[theorem]{Lemma}       
\theoremstyle{definition}
\newtheorem{definition}[theorem]{Definition}

\newif\iffinal
\finalfalse

\iffinal
    \newcommand{\todo}[1]{} 
    \newcommand{\an}[1]{}
    \newcommand{\da}[1]{}
    \newcommand{\cb}[1]{}
\else
    \newcommand{\todo}[1]{\textcolor{red}{\textbf{TODO: #1}}}
    \newcommand{\an}[1]{\textcolor{red}{[#1 -AN]}}
    \newcommand{\da}[1]{\textcolor{blue}{[#1 -DA]}}
    \newcommand{\cb}[1]{\textcolor{magenta}{[#1 -CB]}}
\fi

\title{Towards SMT Solver Stability via Input Normalization}
\author{
\IEEEauthorblockN{
Daneshvar Amrollahi\IEEEauthorrefmark{1},
Mathias Preiner\IEEEauthorrefmark{1},
Aina Niemetz\IEEEauthorrefmark{1},
Andrew Reynolds\IEEEauthorrefmark{2},\\
Moses Charikar\IEEEauthorrefmark{1},
Cesare Tinelli\IEEEauthorrefmark{2},
Clark Barrett\IEEEauthorrefmark{1}
}
\IEEEauthorblockA{\IEEEauthorrefmark{1}Stanford University, USA\\
\texttt{\{daneshvar, preiner, niemetz, moses, barrett\}@cs.stanford.edu}}
\IEEEauthorblockA{\IEEEauthorrefmark{2}The University of Iowa, USA\\
\texttt{\{andrew-reynolds, cesare-tinelli\}@uiowa.edu}}
}

\newcommand{\assertion}{\ensuremath{\alpha}\xspace}
\newcommand{\seqassertion}{A\xspace}

\newcommand{\R}{\ensuremath{\mathbb{R}}\xspace}

\DeclareMathAlphabet{\mathpzc}{OT1}{pzc}{m}{it}
\newcommand{\N}{\ensuremath{\mathpzc{N}}\xspace}

\newcommand{\role}[2]{\ensuremath{\mathit{role}(#1, #2)}\xspace}

\newcommand{\pat}[1]{\ensuremath{\mathit{P}(#1)}\xspace}
\newcommand{\superpattern}[2]{\ensuremath{\mathit{SP}(#1, #2)}\xspace}
\newcommand{\EC}{\ensuremath{\mathit{EC}}\xspace}

\newcommand{\uservar}[1]{\ensuremath{X_{#1}}\xspace}

\newcommand\cpp{{C\nolinebreak[4]\hspace{-.05em}\raisebox{.4ex}{\tiny\bf++}}\xspace}

\newcommand{\benchsmtlib}{\textit{smtlib}\xspace}
\newcommand{\benchmariposa}{\textit{mariposa}\xspace}
\newcommand{\benchdice}{\ensuremath{\text{Dice}_F^*}\xspace}
\newcommand{\benchverid}{\ensuremath{\text{VeriBetrKV}_D}\xspace}
\newcommand{\benchveril}{\ensuremath{\text{VeriBetrKV}_L}\xspace}
\newcommand{\benchkomodod}{\ensuremath{\text{Komodo}_D}\xspace}
\newcommand{\benchkomodos}{\ensuremath{\text{Komodo}_S}\xspace}
\newcommand{\benchwasm}{\ensuremath{\text{vWasm}_F}\xspace}

\begin{document}

\maketitle


\begin{abstract}

In many applications, SMT solvers are utilized to solve similar or identical tasks
over time.  
Significant variations in performance due to small changes in the input
are not uncommon and lead to frustration for users.  This 
sort of \emph{stability} problem represents an important usability
challenge for SMT solvers.
We introduce an approach for mitigating the stability problem based on normalizing solver
inputs.  We show that a perfect normalizing algorithm exists but is
computationally expensive.  We then describe an approximate algorithm
and evaluate it on a set of benchmarks from related
work, as well as a large set of benchmarks sampled from
SMT-LIB.  Our evaluation shows that our approximate normalizer
reduces runtime variability with minimal overhead and is
able to normalize a large class of mutated benchmarks to a unique normal form.

\end{abstract}


\section{Introduction}

SMT solvers are used to solve a large variety of problems in academia and
industry~\cite{DBLP:series/natosec/Bjorner16,DBLP:series/txtcs/KroeningS16,DBLP:conf/cav/Rungta22}.
As these solvers are integrated into more and more workflows, they are
increasingly used in situations where there are many calls to a solver with
identical or similar queries.
For example, a software verification tool may run a
regression suite nightly to check that some software meets its
specification, and for the most part, this nightly run does not differ much, if at
all, from the previous night's run.
A common pain point in such situations is that SMT solver performance can vary
significantly, even if there are only minor changes.  This has been termed
the \emph{stability} problem: queries that are semantically similar
or identical may require vastly different amounts of time to solve.  Or, even
worse, some minor changes may result in a formerly solved query not being
solved at all.

As part of an NSF-supported project~\cite{nsfstudy}, we
spoke with a large number of SMT stakeholders and cataloged and ranked over 100
different issues, based on how often
they came up in interviews.  In this ranking, instability was ranked as the second-highest area
of concern (behind only requests for better diagnostic output when the solver
is unable to solve a problem).  The instability concern was also highlighted in a
keynote talk by Neha Rungta at CAV 2022~\cite{DBLP:conf/cav/Rungta22} on
the use of SMT solvers at Amazon Web Services
and
has been identified by another senior manager at Amazon as one of two top
priorities for improving their SMT solving workflows~\cite{JonesPrivate}.

There are two kinds of changes that can introduce instability: changes to the
input and changes to the solver.  While both are important, this paper focuses on the first:
\emph{changes to the input}.  In other words, we are interested in
reducing the sensitivity of solving time to minor changes in an input
formula, especially if those changes are semantics-preserving.


Two common misconceptions about stability are worth noting.  First,
instability should not be confused with poor solver performance.
Improving solver performance is an important and
worthy goal; however, it is orthogonal to addressing the stability problem.
Our focus is on making solver
performance more \emph{consistent} in the face of mutations to its input.  Importantly, many users in our study reported that
they would welcome \emph{any} improvement to stability, even if it came at the
cost of a degredation in performance.
The other misconception is the failure to recognize that
the stability problem is linked to computational complexity.  Because the SMT
problem is NP-hard (or worse, depending on the theory being used), solvers try to guide worst-case exponential algorithms in such
a way that solutions are found quickly when possible.  However, even a small change
in the input can trigger a different search path, which could result in
an exponentially worse (or better) runtime.  This is similar to the well-known
``butterfly effect'' in chaos theory.
Despite this, as we show in this paper, there \emph{are} steps that can be taken to
\emph{reduce} instability.


SMT solvers typically already use deterministic data
structures and algorithms to eliminate easily avoidable sources of instability.
In a sequential execution setting,
the remaining source of variability is their
heuristics, primarily their branching heuristics, which typically break ties
based on the \emph{order} in which declarations and assertions appear in
the input formula.
Our key idea is to reduce this variability by using the \emph{structure of the input problem}
to determine the order of declarations and assertions.
Our approach uses the principle of \emph{normalization}: we attempt
to map classes of semantically equivalent inputs to the same normal form.
Note that if perfect normalization could be achieved in this setting, it would result in perfect
stability,
with respect to the semantically equivalent input classes,
as the input to the core solver would be the same in every case.
We present an approach aimed at getting closer to that ideal
without introducing too much overhead.
A notable feature of our approach is that it is agnostic to the underlying solver.  It
can thus be used to improve the stability of any SMT solver with respect to
changes in the input.

We consider normalization with respect to
a set of basic, semantics-preserving transformations:
($i$) reordering of assertions;
($ii$) reordering of operands of commutative operators;
($iii$) reordering and renaming of user-defined symbol declarations; and
($iv$) replacing anti-symmetric operators 
     by their converse.
These transformations are representative of the kinds of changes that happen in
the real-world scenarios motivating our work.
For example, in a program verification workflow, if a function is moved or
renamed, this could result in a reordering of declarations and
assertions, a different name for a user-defined symbol, or both.
%
These transformations are also a superset of the
solver input transformations considered in previous work on \emph{measuring}
instability~\cite{Zhou2023}.  We will refer to these transformations as
\emph{mutations} and to inputs that have been transformed this way as
\emph{mutated} inputs.
We address the following research questions:

\begin{enumerate}
    \item Is it possible to design a normalizing algorithm that utilizes
      these mutations to map all mutated variants
      to a single unique normal form?
    
    
    \item If such an algorithm exists, what is its time complexity?
      
    
    \item How closely can an efficient algorithm approximate the ideal algorithm?
    
\end{enumerate}


After covering some background in \Cref{sec:background}, we formalize 
the problem in \Cref{sec:formal} and
answer the first two questions, showing that
such an algorithm does indeed exist but is as hard as graph isomorphism.
We provide an answer to the third question in the remainder of the paper.
\Cref{sec:methodology} introduces
an algorithm that approximates the ideal algorithm, and
\Cref{sec:evaluation} presents an
evaluation of our implementation, showing that it significantly improves SMT solver stability on benchmarks from~\cite{Zhou2023} and
from the SMT-LIB benchmark library~\cite{smtlib2024}.  Finally, \Cref{sec:conclusion} concludes.

\noindent
\textbf{\emph{Related work.}}
The importance of the issue of stability in SMT solving has been raised in other
work~\cite{galois,komodo,ironclad,clement}.
Dodds~\cite{galois}, highlights the problem of proof fragility under
changes in verification tools.
Ferraiuolo et al.~\cite{komodo} mention
proof instability as \emph{the most frustrating recurring problem}, especially
when proof complexity increases as a result of reasoning about procedures with many
instructions and complex specifications.
In Hawblitzel et al.~\cite{ironclad}, verification instability is observed in large formulas and
non-linear arithmetic due to different options for applicable heuristics.
Leino et al.~\cite{clement} identify
\emph{matching loops}---caused by 
certain forms of quantifiers
that lead an SMT solver
to repeatedly instantiate a limited set of quantified formulas---as a key factor
contributing to instability in verification times, and describe
techniques to detect and prevent them.

More relevant to our work is the work of Zhou et al.~\cite{Zhou2023,shake}.
In~\cite{Zhou2023}, they pioneer an effort to detect and quantify instability and introduce a tool for this task called Mariposa. 
They show that mainstream SMT solvers such as Z3~\cite{z3} and cvc5~\cite{cvc5}
exhibit instability on a set of F\(^{*}\)~\cite{fsharp} and Dafny~\cite{dafny}
benchmarks~\cite{vwasm,komodo,veribetrkv,veribetrkvl,serval,dicesharp}. 
They consider the benchmark-modifying mutations of symbol renaming
and assertion reordering, as well as solver-modifying mutations via the use of
different random seeds, and use a statistical approach to identify instability
arising from these mutations.  We include an evaluation of our technique on
their benchmarks, with the notable difference that we do not
include solver-modifying mutations.  This is because we
aim to improve stability with respect to input changes, rather than
%
across changes to a solver.
In~\cite{shake}, Zhou et al. identify irrelevant context in a query
as one source of instability and propose a novel approach to filter out
such context to improve solver stability.  This approach is
complementary to our own, and combining the two is an interesting direction
for future work.

Also closely related is the work of Weber~\cite{weber16scrambling}, 
which introduces a normalizer
designed to reverse the effects of mutations used to \emph{scramble} benchmarks for
the SMT-COMP competition at the time.
However, the normalizer introduced in~\cite{weber16scrambling}
was specifically designed to expose a weakness in the scrambling
algorithm, namely that symbols were renamed but not reordered.  This weakness
makes it possible to achieve perfect normalization efficiently, which their
algorithm does.  However, in the presence of the more expressive and realistic
mutations we consider, their normalizer performs poorly, as we show in \Cref{sec:evaluation}.



\section{Background}
\label{sec:background}

\newcommand{\boolS}{\texttt{Bool}\xspace}
\newcommand{\usersymbols}{\ensuremath{\mathit{user}}\xspace}


We work in the context of many-sorted logic (e.g.,
\cite{enderton2001mathematical}), where we assume an infinite set of variables
of each sort and the usual notions of signatures, terms, formulas, assignments, and
interpretations.
We assume a signature $\Sigma$ consisting of sort symbols and sorted function symbols.
It is convenient to assume that $\Sigma$ has a distinguished sort \boolS,
for the Booleans, and to represent relation symbols as function symbols 
whose return sort is \boolS.  We also assume the signature includes equality.
Symbols in $\Sigma$ are partitioned into \emph{theory symbols} (e.g.,
$=,\land,\lor,+,-,0,1$) and \emph{user-defined symbols} (e.g., $f$, $g$, $x$,
$y$).  We assume some \emph{background theory} restricts the 
theory symbols to have fixed interpretations, whereas the interpretation of
user-defined symbols is left unrestricted.

We represent formulas as finite \emph{sequences} of symbols in
prefix notation, where each symbol is either a theory symbol or a user-defined
symbol. 
This causes no ambiguity when each symbol has a fixed arity.
If $S$ is a sequence $\langle s_1,\dots,s_n\rangle$, we write $|S|$ to
denote $n$, the length of the sequence, and $ S_i $ to denote the $i^{\mathit{th}}$ element of the sequence.  We write $s\in S$ to mean that $s$
occurs in the sequence $S$, and write $S \circ S'$ for the sequence obtained by
appending $S'$ at the end of $S$.
We write $\usersymbols(S)$ to mean the sequence obtained by
deleting all theory symbols in $S$, e.g., $\usersymbols(\langle
x,+,y,-,x\rangle)=\langle x,y,x \rangle$.  We denote the set of integers
between $m$ and $n$ inclusive, where $n \ge m$, as $[m,n]$.  And we abbreviate
$[1,n]$ as $[n]$.


\begin{figure}[t]
        \centering
        \begin{minipage}{0.45\textwidth}%
          \begin{lstlisting}[basicstyle=\ttfamily\scriptsize, frame=single]
(set-logic QF_UFLIA)

(declare-fun f (Int) Int)  (declare-const v Int) 
(declare-const w Int)      (declare-const x Int) 
(declare-const y Int)      (declare-const z Int)

(assert (>= (+ (f x) y) (- v 12)))
(assert (< (+ x y) (* x z)))
(assert (>= (+ (f y) x) (- w 12)))
(assert (< (+ y x) (* x x)))
(assert (< (+ x y) (* y y)))
(assert (< (+ y x) (* y v)))

(check-sat)
          \end{lstlisting}
        \end{minipage}
\vspace*{-.5em}
        \caption{Running Example.}
\vspace*{-.5em}        
        \label{fig:motivating-example}
\end{figure}

\begin{figure}[t]
    \centering
    \begin{minipage}{0.45\textwidth}%
            \begin{lstlisting}[basicstyle=\ttfamily\scriptsize, frame=single]
(assert (>= (+ (f x) y) (- v 12)))
(assert (< (+ x y) (* x z)))
            \end{lstlisting}
    \end{minipage}
    \begin{minipage}{0.45\textwidth}%
            \begin{lstlisting}[basicstyle=\ttfamily\scriptsize, frame=single]
(assert (> (* u4 u2) (+ u2 u3)))
(assert (>= (+ (g u2) u3) (- u1 12)))
            \end{lstlisting}
    \end{minipage}%
\vspace*{-.5em}
    \caption{Two different mutated versions of the same assertions.}
    \vspace*{-1.5em}
    \label{fig:scramble-example}
\end{figure}
A formula in the SMT-LIB 2.6 format~\cite{BarFT-SMTLIB} is shown in
\Cref{fig:motivating-example}.  We use this as a running example throughout the paper.
The example includes arithmetic theory symbols, a
user-defined function $f$, and user-defined constants $v$, $w$, $x$, $y$, and
$z$.  We will refer to the sequence of six assertions in the running example as
$\langle \beta_1, \dots \beta_6 \rangle$.
\Cref{fig:scramble-example} shows two different representations of the
first two assertions in the example: the first is from the example and the second is a mutation of it.  In particular, the order of the assertions has been swapped, the
operands of the \texttt{*} operator have been reordered, the user-defined symbols have been
renamed (with \texttt{w}, \texttt{x}, \texttt{y}, and \texttt{z}
renamed to \texttt{u1} through \texttt{u4},
respectively, and \texttt{f} renamed to \texttt{g}), and the assertion using
\texttt{<} has been rewriiten to use \texttt{>}.


\section{Formalization}
\label{sec:formal}

Consider again the first two assertions from the example.
Here, the user-defined symbols are \{$f$, $x$, $y$, $v$, $z$\}, and the theory
symbols are \{$>=$, $<$, $+$, $-$, $*$, $12$\}.
When written as a sequence in prefix
notation, $\beta_1$ is $\langle >=,+,f,x,y,-,v,12\rangle$.
Similarly, $\beta_2$ is $\langle <,+,x,y,*,x,z \rangle$.

Recall that mutations include four operations:
($i$) reordering
assertions; ($ii$) reordering operands of commutative operators; ($iii$)
reordering and renaming symbols; and
($iv$) replacing anti-symmetric operators.  We defer ($iv$) to the end of this section,
as it can easily be handled separately.  To formalize the others, we introduce some definitions.

\begin{definition}[Shuffle Set]
  The \emph{shuffle set} of a sequence $\seqassertion$ is defined as
  $S(\seqassertion) = \{\seqassertion' \mid \seqassertion'$ is a
    permutation of $\seqassertion\}$.
\end{definition}

\noindent
For example, $S(\langle \beta_1,\beta_2 \rangle) = \{ \langle \beta_1,\beta_2 \rangle, \langle \beta_2,\beta_1 \rangle\}$.

\begin{definition}[Commutative Reordering Set]
Let $\assertion$ be a formula, possibly containing (binary) commutative operators.  The
\emph{commutative reordering set} of $\assertion$ is
defined as $C(\assertion) = \{\assertion' \mid \assertion'$ is the result of
swapping the operands of zero or more of these commutative operators in $\assertion\}$.  For a sequence $\seqassertion$
of $n$ formulas, $C(\seqassertion)$ is the set of all sequences $\langle \assertion'_1,\dots,\assertion'_n\rangle$ where for each $i\in[n]$, $\assertion'_i \in C(\assertion_i)$.
If $X$ is a set (containing either formulas or sequences of formulas), then $C(X)=\{x' \mid x' \in C(x) \text{ for some } x\in X\}$.
\end{definition}

\noindent
For example, $C(\beta_1) = \{\beta_1, \langle >=,+,y,f,x,-,v,12\rangle\}$.  Note that
there is only one entry in $C(\beta_1)$ besides $\beta_1$ itself, because $+$ is the
only commutative operator appearing in $\beta_1$.  On the other hand, there are
four elements in $C(\beta_2)$, since $\beta_2$ has two commutative operators, $+$ and $*$.
Consequently, $C(\langle \beta_1, \beta_2 \rangle)$ has eight elements, representing all combinations of
one formula each from $C(\beta_1)$ and $C(\beta_2)$.


\begin{definition}[Pattern]
  \label{def:pattern}
  Let $\assertion$ be a formula.  The
  \emph{pattern} of $\assertion$, written $\pat{\assertion}$, is a sequence of the
  same length as $\assertion$, defined for each $i\in[|\assertion|]$ by
  \[\vspace*{-.3em}
    \pat{\assertion}_i =
    \begin{cases}
      \assertion_i
        & \text{if $\assertion_i$ is a theory symbol,}\\[1ex]
      @1
        & \begin{aligned}[t]
            &\text{if $\assertion_i$ is the first user-defined symbol}\\
            &\text{appearing in $\assertion$,}
          \end{aligned}\\[1ex]
      \pat{\assertion}_j
        & \begin{aligned}[t]
            &\text{if $\assertion_i$ is a user-defined symbol and}\\
            &\text{there exists }j\in[i-1]\text{ with }\assertion_j=\assertion_i,
          \end{aligned}\\[1ex]
        @k
          & \begin{aligned}[t]
              &\text{otherwise, where } k = 1\ \ +\\
              &\bigl|\{\assertion_j\mid j\in[i-1],\,\assertion_j\text{ user-defined}\}\bigr|.
            \end{aligned}
    \end{cases}
  \]
\end{definition}

\noindent
For convenience, we assume that each symbol $@k$ is a fresh
constant\footnote{The SMT-LIB 2.6 standard reserves symbols starting with $@$ for
internal use by solvers, so this assumption is a reasonable one.}
of the same sort as the symbol it is replacing, so that if $\assertion$ is well-sorted, then so is \pat{\assertion}.
We call introduced symbols starting with ``@'' \emph{pattern symbols} and
sequences in the co-domain of $P$ \emph{patterns}.
We define a total order $\prec$ on patterns to be the lexicographic order induced
by some total order on formula symbols.\footnote{An obvious choice for this
order (and the one we use) is the lexicographic order on the string
representations of theory and pattern symbols.}
For our example assertions, we have $\pat{\beta_1} = \langle >=,+,@1,@2,@3,-,@4,12\rangle$
and $\pat{\beta_2} = \langle <,+,@1,@2,*,@1,@3\rangle$.

We lift this notation to sequences and sets of sequences.  To explain how, we need two more
definitions.

\newcommand{\Conj}{\ensuremath{\mathit{Conj}}\xspace}
\newcommand{\Unconj}{\ensuremath{\mathit{Unconj}}\xspace}

\begin{definition}
  Given a sequence of formulas $\seqassertion = \langle
  \assertion_1,\dots,\assertion_n\rangle$, the
  \emph{conjoining} of $\seqassertion$, $\Conj(\seqassertion)$  is the formula
  $\langle \wedge \rangle \circ \assertion_1 \circ \assertion_2 \circ \dots
  \circ \assertion_n$.  Similarly, given a formula
  $\langle \wedge \rangle \circ \assertion_1 \circ \assertion_2 \circ \dots
  \circ \assertion_n$, where $\assertion_i$ is a
  formula for $i\in[n]$, the \emph{unconjoining} of $\alpha$, written
  $\Unconj(\alpha)$ is the formula sequence $\langle
  \assertion_1,\dots,\assertion_n\rangle$.
\end{definition}

\noindent
For a sequence $\seqassertion$ of formulas, we define
$\pat{\seqassertion} = \Unconj(\pat{\Conj(\seqassertion)})$.
If $X$ is a set (containing either formulas or sequences of formulas), then $\pat{X}=\{\pat{x} \mid x \in X\}$.
We similarly lift $\prec$ to sequences of formulas: if $\seqassertion$ and $\seqassertion'$
are sequences of patterns, then $\seqassertion \prec \seqassertion'$ iff $\assertion_1 \circ \dots \circ \assertion_{|\seqassertion|}
\prec \assertion'_1 \circ \dots \circ \assertion'_{|\seqassertion'|}$.
We next define renaming and normalization.

\begin{definition}[Renaming]
Let $V$ be a set of variables.
A \emph{renaming} $R$ is an injective function from pattern symbols to
  $V$.  For a formula $\assertion$, $R(\assertion)$ is defined to be a
sequence of the same size as $\assertion$, defined as follows:
\[
R(\assertion)_i =
\begin{cases}
  \assertion_i & \text{if } \assertion_i \text{ is a theory symbol,}\\
  R(\assertion_i) & \text{if } \assertion_i \text{ is a pattern symbol.}
\end{cases}
\]
For a sequence of formulas $\seqassertion=\langle
\assertion_1,\dots,\assertion_n\rangle$, $R(\seqassertion) = \langle R(\assertion_1),\dots,R(\assertion_n)\rangle$.
If $X$ is a set (containing either formulas or sequences of formulas), then
$R(X)=\{R(x) \mid x \in X\}$.\footnote{
When representing an assertion sequence as an SMT-LIB 2.6 script, we declare user-defined symbols in lexicographic order, so a
renaming can affect the order of declarations.}
\end{definition}

\noindent




\begin{definition}[Normalizing Function]
A function $N$ from sequences of formulas to sequences of formulas is said to be \emph{normalizing} if, for every sequence $\seqassertion$ of formulas:
\begin{enumerate}
  \item $N(\seqassertion) = R(\seqassertion')$ for some $\seqassertion' \in \pat{C(S(\seqassertion))}$ and some renaming $R$; and
  \item if $M_1= R_1(M'_1)$ and $M_2= R_2(M'_2)$, with $R_1,R_2$ renamings and with
    $M'_1,M'_2 \in \pat{C(S(\seqassertion))}$, then $N(M_1) = N(M_2)$.
\end{enumerate}
\end{definition}

\noindent
We can now introduce our first research question: does there exist a normalizing function?  We
show that the question can be answered affirmatively.

\begin{definition}[Normalizing Function]
Let $\N$ be defined as follows.  Given a sequence of formulas $\seqassertion$, let $\N(\seqassertion)$
be the $\prec$-minimal element of $\pat{C(S(\seqassertion))}$.
\end{definition}





\begin{theorem} \label{thm:normalizing}
  Function $\N$ is normalizing. 
\end{theorem}

\noindent
Due to lack of space, proofs for this and other theorems in this paper are
omitted.\footnote{For the reviewers' convenience, the proofs are provided in an appendix.}

\subsection{Complexity}

We have shown the existence of a normalizing function.  The next question
is whether normalization can be done efficiently.  An informal argument that
computing the normalization of an arbitrary set of formulas is at least as hard
as graph isomorphism can be found in~\cite{stackexchange}, and a formal proof
can be found in~\cite{weber16scrambling}.
The question of
whether graph isomorphism can be solved in polynomial time is a long-standing
open problem.

\begin{theorem}\label{thm:normGI}
  Let N be a normalizing function.  Then, computing $N(\seqassertion)$ for an arbitrary
  $\seqassertion$ is as hard as solving graph isomorphism.
\end{theorem}

\noindent
The use of $C$ is not essential for the proof to succeed.
Normalizing just the result of shuffling and renaming is already as hard as graph isomorphism.

\subsection{Anti-symmetric Operators}

%
As mentioned above, mutations can randomly replace anti-symmetric operators with their dual
operator. For example, \texttt{(< (+ x y) (* x z))} could be changed to
\texttt{(> (* x z) (+ x y)
)}. In general, we allow mutations that transform expressions of the form
\( A \; op \; B \) into \( B \; op' \; A \), where (\(op \), \(op'\)) pairs
include: (\texttt{>}, \texttt{<}), (\texttt{>=}, \texttt{<=}),
(\texttt{bvugt}, \texttt{bvult}), (\texttt{bvuge}, \texttt{bvule}), \ldots.


A normalization algorithm can easily handle anti-symmetric operators, simply by
choosing one representative operator for each pair and forcing all assertions to use only the chosen
operators.  For example, if the first operator in each pair
listed above is the chosen one, then \Cref{fig:antisymm} shows the result
of normalizing the first two assertions in our running example.
Notice that the second assertion is modified by replacing \texttt{<} with
\texttt{>} and swapping the operands but the first assertion is unchanged
because it is already using the chosen operator.

\begin{figure}[t]
    \begin{minipage}{0.48\textwidth}
        \centering
        \begin{lstlisting}[basicstyle=\ttfamily\scriptsize, frame=single]
(assert (>= (+ (f x) y) (- v 12)))
(assert (< (+ x y) (* x z)))
        \end{lstlisting}
    \end{minipage}
    \hfil
    \begin{minipage}{0.48\textwidth}
        \begin{lstlisting}[basicstyle=\ttfamily\scriptsize, frame=single]
(assert (>= (+ (f x) y) (- v 12)))
(assert (> (* x z) (+ x y)))
        \end{lstlisting}
    \end{minipage}
    \vspace*{-.5em}
    \caption{Original (top) and normalized (bottom) assertions.}
    \vspace*{-1.5em}
    \label{fig:antisymm}
\end{figure}

\section{Approximating a Normalization Algorithm} \label{sec:methodology}



As a first step towards a general practical algorithm for normalization, we
describe a heuristic procedure designed to handle two of the four mutations:
shuffling and renaming.  We leave the handling of the other operations to
future work.  We expect support for normalizing antisymmetric operator
replacement to be straightforward (as described above), while normalizing
commutative operand swapping will be more challenging.  Note that shuffling and
renaming are also the two operations used to mutate benchmarks in the Mariposa
work (the closest related work)~\cite{Zhou2023}.  We describe our algorithm at
a conceptual level here, and discuss several optimizations necessary to make it
work well in practice in~\Cref{sec:optimization}.
Our algorithm consists of three steps: ($i$) Sorting the assertions; ($ii$)
Renaming all symbols; and ($iii$) Sorting the assertions again.
In the rest of this section, we discuss these steps in detail.

\noindent
\textbf{\emph{Sorting the assertions.}}
Step one is to sort the assertions.  The challenge is to do this in a way that does
not depend on the names of user-defined symbols.
The key idea is to use patterns.  In particular, to order assertions $\alpha$ and
$\alpha'$, we can compare \pat{\alpha} and \pat{\alpha'} using the $\prec$ order.
For instance, considering agian assertions from our running example,
$\beta_5$ is $\langle <,+,x,y,*,y,y\rangle$, and $\beta_6$ is $\langle
<,+,y,x,*,y,v\rangle$, so we have
$\pat{\beta_5} = \langle <, +, @1, @2, *, @2, @2 \rangle$ and $\pat{\beta_6} = \langle <, +, @1,
@2, *, @1, 3 \rangle$. The first difference in the patterns is in the sixth position.
Assuming $@1$ is ordered before $@2$, we have
that $\pat{\beta_6} \prec \pat{\beta_5}$, so we can conclude that $\beta_6$ should
be placed before~$\beta_5$.

Note, however, that it is possible for two formulas to have the same
pattern.  Thus, after sorting according to patterns, we obtain
an ordered list of equivalence classes $\EC_1, \cdots, \EC_n$ with the following features: 
\begin{enumerate}
    \item $\alpha$ and $\alpha'$ belong to the same equivalence class iff $\pat{\alpha} = \pat{\alpha'}$.
    \item $\alpha \in \EC_i$ and $\alpha' \in \EC_j$ where $i < j$ iff $\pat{\alpha} < \pat{\alpha'}$.
\end{enumerate}

The next question is whether we can easily order the assertions belonging to
the same equivalence class.  We give an efficient approximation method.
For this, we need the notions of role and super-pattern.

\begin{definition}[Role]
  \label{def:role}
  The role of a symbol $s$ in a formula \assertion, denoted \role{s}{\assertion}, is $0$ if
    $s\not\in \assertion$ and the index of the earliest occurrence of $s$ in
    $\usersymbols(\assertion)$, otherwise.  The role of $s$ in a set of formulas is the
    multiset consisting of all the roles played by $s$ in the formulas in the set.
\end{definition}
\noindent
For example, consider the role of $y$ in
$\beta_5$.  First of all, we compute
$\usersymbols(\beta_5)$, which is $\langle x,y,y,y \rangle$.  We can then see
that $y$ occurs first at the second position, so $\role{y}{\beta_5}=2$.
Similarly, $\role{x}{\{\beta_4,\beta_5,\beta_6\}} = \{1, 1, 2\}$.


\begin{definition}[Super-pattern]
  \label{def:super-pattern}
    The super-pattern of a symbol $s$ over a sequence $X$ of sets $X_1, \ldots,
    X_{n}$, denoted $\superpattern{s}{X}$, is the sequence of roles of the symbol in each set:
    $\superpattern{s}{X} = \langle \role{s}{X_1}, \role{s}{X_2}, \dots, \role{s}{X_n} \rangle$.
\end{definition}

\noindent
Let $\EC$ be a sequence of formula equivalence classes.  The super-pattern of
$\EC$ captures the role of a symbol across all equivalence classes, while
treating the formulas in each equivalence class as unordered.

\begin{figure}[t]
    \centering
    \begin{minipage}{0.45\textwidth}%
      \begin{lstlisting}[basicstyle=\ttfamily\scriptsize, frame=single]
(assert (< (+ x y) (* x z)))
(assert (< (+ y x) (* y v)))

(assert (< (+ y x) (* x x)))
(assert (< (+ x y) (* y y)))

(assert (>= (+ (f x) y) (- v 12)))
(assert (>= (+ (f y) x) (- w 12)))
      \end{lstlisting}
    \end{minipage}
    \vspace*{-.5em}
    \caption{Assertions sorted by pattern.}
    \vspace*{-1.5em}
    \label{fig:motivating-example-eqclass}
\end{figure}

To illustrate, recall the example from
\Cref{fig:motivating-example}. \Cref{fig:motivating-example-eqclass}
shows the result of sorting the assertions by pattern, resulting in three equivalence classes, each
separated by an empty line.
The patterns of the equivalence classes in $\EC = \{\EC_1, \ldots, \EC_3\}$,
from top to bottom, are as follows:
\begin{itemize}
  \item[] $\EC_1$: $\langle <, +, @1, @2, *, @1, @3 \rangle$
  \item[] $\EC_2$: $\langle <, +, @1, @2, *, @2, @2 \rangle$
  \item[] $\EC_3$: $\langle >=, +, @1, @2, @3, -, @4, 12 \rangle$
\end{itemize}

\noindent
Now, suppose we want to order the formulas in $\EC_3$. We compare the
super-patterns of the first different pair of user-defined symbols, in this
case, $x$ and $y$.
The roles of $x$ throughout the equivalence classes are:
\begin{itemize}
  \item[] $\role{x}{\EC_1} = \{ \role{x}{\beta_2}, \role{x}{\beta_6} \} = \{ 1, 2 \}$
  \item[] $\role{x}{\EC_2} = \{ \role{x}{\beta_4}, \role{x}{\beta_5} \} = \{ 1, 2 \}$
  \item[] $\role{x}{\EC_3} = \{ \role{x}{\beta_1}, \role{x}{\beta_3} \} = \{ 2, 3 \}$
\end{itemize}

\noindent
Therefore, applying the definition of super-pattern for $x$ yields
$\superpattern{x}{\EC} =
\langle \{ 1, 2 \}, \{ 1, 2 \} , \{ 2, 3 \}\rangle$.  It is not hard to see
that the super-pattern for $y$ is the same, so the two assertions cannot be
distinguished by looking at $x$ and $y$.  The next pair of different user-defined
variables also consists of $x$ and $y$.  However, the last pair is $v$ and $w$.
Following the same process, we find that
$\superpattern{v}{\EC}=\langle\{0,4\},\{0,0\},\{0,4\}\rangle$ and
$\superpattern{w}{\EC}=\langle\{0,0\},\{0,0\},\{0,4\}\rangle$.  Now, all we
need is a way to order different super-patterns.

\begin{definition}[Integer multiset order]
  Given multi-sets of integers $m_1$ and $m_2$, $m_1 < m_2$ iff the sequence of
  nondecreasing elements of $m_1$ is lexicographically less than the sequence of
  nondecreasing elements of $m_2$.
\end{definition}

\begin{definition}[Super-pattern order]
  For super-patterns $s_1$ and $s_2$, $s_1 < s_2$ iff $s_1$ comes before $s_2$
  when compared using the lexicographic order induced by the integer multiset order.
\end{definition}

\noindent
Thus, when comparing super-patterns, we compare the entries in the sequences one by
one using the integer multiset order. The
first two entries in $\superpattern{v}{\EC}$ and $\superpattern{w}{\EC}$ are
$\{0,4\}$ for $v$ and $\{0,0\}$ for $w$.
Because $\langle 0,0\rangle < \langle 0,4\rangle$, we can conclude that $\superpattern{w}{\EC} < \superpattern{v}{\EC}$.
Thus, we should switch the order of assertions in the last equivalence class.
Similarly, by computing super-patterns for $z$ and $v$, we can see that we
should keep the order of the assertions in the first equivalence class.

It is possible for all of the super-patterns of corresponding symbols in two assertions in
the same equivalence class to be the same.  In this case, our heuristic
algorithm fails, and the assertion order is left unchanged.  For example, for
$\EC_2$, the only user-defined symbols available for comparison are $x$ and
$y$, and they have the same super-pattern.  Thus, we leave these assertions
in their original order for now.
\Cref{alg:assertion-ordering} shows the full algorithm for ordering two assertions.

\begin{algorithm}
    \caption{Algorithm for ordering assertions $A$ and $B$ given $\EC$, a
      sequence of equivalence classes (with respect to pattern equality) of assertions.}
    \label{alg:assertion-ordering}
    \begin{algorithmic}
\IF{$\pat{A} \neq \pat{B}$}
    \RETURN $\pat{A} \prec \pat{B}$
\ENDIF

\FOR{$i \gets 1$ to $|\usersymbols(A)|$}
    \STATE $u \gets \usersymbols(A)_i$
    \STATE $v \gets \usersymbols(B)_i$
    \IF{$u = v$}
        \STATE \textbf{continue}
    \ENDIF
    \IF{$\superpattern{u}{\EC} = \superpattern{v}{\EC}$}
        \RETURN $\superpattern{u}{\EC} < \superpattern{v}{\EC}$
    \ENDIF
\ENDFOR
\RETURN \textbf{inconclusive}
\end{algorithmic}
\end{algorithm}

\noindent
\textbf{\emph{Renaming all symbols.}}
After sorting the assertions according to \Cref{alg:assertion-ordering}, we
rename all the symbols in the assertions.  We use a renaming $\R$ that maps a
variable whose pattern symbol is $@k$ to $X_k$.  More
precisely, if $\seqassertion$ is the sequence of assertions after sorting, we
replace $\seqassertion$ with $\R(\pat{\seqassertion})$.
\Cref{fig:after-renaming} shows the assertions from our running example
after sorting and renaming.

\begin{figure}[t]
    \centering
    \begin{minipage}{0.45\textwidth}%
      \begin{lstlisting}[basicstyle=\ttfamily\scriptsize, frame=single,escapechar=|]
(assert (< (+ |\uservar{1}| |\uservar{2}|) (* |\uservar{1}| |\uservar{3}|)))
(assert (< (+ |\uservar{2}| |\uservar{1}|) (* |\uservar{2}| |\uservar{4}|)))

(assert (< (+ |\uservar{2}| |\uservar{1}|) (* |\uservar{1}| |\uservar{1}|)))
(assert (< (+ |\uservar{1}| |\uservar{2}|) (* |\uservar{2}| |\uservar{2}|)))

(assert (>= (+ (|\uservar{5}| |\uservar{1}|) |\uservar{2}|) (- |\uservar{4}| 12)))
(assert (>= (+ (|\uservar{5}| |\uservar{2}|) |\uservar{1}|) (- |\uservar{6}| 12)))
      \end{lstlisting}
    \end{minipage}
    \vspace*{-.5em}
    \caption{Assertions after sorting and renaming.}
    \vspace*{-1.5em}
    \label{fig:after-renaming}
\end{figure}

\noindent
\textbf{\emph{Sorting the assertions again.}}
After renaming, there is one more step that can improve the normalizer.  It is
based on the observation that within an equivalence class, different
assertion orders are possible, depending on the initial order,
when all symbols in a pair of assertions have the same super-patterns.  This can be
partially addressed by lexicographically sorting each equivalence
class after renaming.
This ensures that if we have two benchmarks for which the first two steps produce the same set of
assertions, but in different orders, then these two benchmarks will be
normalized the same way.

Looking again at \Cref{fig:after-renaming}, we see that assertions in the first
and last equivalence classes are already in sorted order.  However, the
assertions in equivalence class 2 should be reordered.  Recall that in the
previous step, we did not have a way to order these assertions, but now there
is an unambiguous order for them.
It is important to note that our algorithm does not guarantee the
normalization property.  The reason for this incompleteness is that when
assertions cannot be distinguished by super-patterns, there can be different
normal forms for benchmarks, even if one is a mutation of the other (via
shuffling and renaming).  However, our algorithm works well in
practice, as we show next.





\section{Experiments}
\label{sec:evaluation}


We implemented our normalization algorithm in \cpp.  We evaluate its
effectiveness on three dimensions: \emph{normalization effectiveness},
\emph{stability}, and \emph{runtime}.

We use two sets of benchmarks.  The first set, \benchsmtlib, is obtained
by randomly selecting 50~benchmarks from each family (or all of them, if there
are fewer than 50) in the SMT-LIB benchmark library~\cite{smtlib2024}.
A family, in this context, consists of all benchmarks in a single
leaf directory of the filesystem heirarchy.
Note that even though our normalization algorithm as described in
\Cref{sec:methodology} is, in theory, applicable to all of SMT-LIB without limitation, our
current implementation does not support the normalization of
algebraic datatypes due to implementation-level complexities.
We thus excluded benchmarks with algebraic datatypes from our evaluation.
Still, \benchsmtlib contains 41,166 benchmarks from 1,581
benchmark families.
We group benchmarks from this set into the (so-called)
divisions used by the annual SMT-COMP~\cite{smtcomp} to keep the number of
categories manageable.

The second set, \benchmariposa,
consists of the benchmarks used in~\cite{Zhou2023} (as provided
in~\cite{mariposazenodo}), originating from program verification projects
written in Dafny~\cite{dafny}, Serval \cite{serval}, and F* \cite{fsharp}.
\benchmariposa contains 16,622 benchmarks and is
divided into six families:
\benchdice~\cite{dicesharp} (1536), \benchkomodod and
\benchkomodos~\cite{komodo} (2,054 and 773),
\benchverid~\cite{veribetrkv} (5,170),
\benchveril \cite{veribetrkvl} (5,334) and
\benchwasm~\cite{vwasm} (1,755).



For each benchmark, we produce 10 mutations by randomly applying assertion
shuffling and renaming.  Our renaming uses a fixed enumeration of names (i.e.,
a benchmark wth $n$ user-defined symbols always uses the first $n$ of these names), but shuffles
the order in which they appear in assertions (and thus the order in which they
are declared, since we declare symbols in the order in which they appear).
We ran all experiments on a cluster of 
48 machines with AMD Ryzen 9 7950X CPUs
using a time limit of 60 seconds and a memory limit of 8GB.  We used this
same time and memory limit to separately limit the mutation step, the normalization step, and
the solving step.\footnote{This is long enough to be able to catch cases when the normalizer is
slow and short enough to help keep a large evaluation computationally
tractable.  Using the same timeout for all stages also keeps things simple.}
We do not report results for benchmarks that timed out during mutation or normalization.
For the \benchsmtlib set, we exclude 6 benchmarks that timed out during the
mutation step and 85 that timed out during the normalization phase (of these,
58 timed out during parsing, before ever getting to the normalization
algorithm), for any of the 10 mutated versions.
In \benchmariposa, we only exclude 2 benchmarks, both of which timed out during
normalization.
Overall, this resulted in 41,075 eligible benchmarks in \benchsmtlib and 16,620
in \benchmariposa.


\subsection{Normalization Effectiveness}\label{sec:uniqueness}

In our first experiment, we evaluate how closely
our implementation of an efficient normalization algorithm approximates the
ideal algorithm on both benchmark sets.
For this, we measure, for each benchmark: ($i$) the number of distinct outputs produced by our
normalizer for the 10 mutations (the best possible is 1; the worst possible is 10); and ($ii$)
the \emph{similarity} of the 10 normalized outputs, computed as the average percentage of
identical lines among all pairs of outputs.
\Cref{tab:unique} shows our results.  Each row starts with a category
name (the name of the division or family) and a pair of numbers ($x/y$) denoting the number of excluded benchmarks $x$
vs.~the total number of benchmarks $y$ in that category.  We then list
the number of considered benchmarks (\textit{\#Bench}, equal to $y-x$) and the average number of
distinct benchmarks produced by the 10 mutations, before normalization (\textit{pre}) and after normalization, without
(\textit{wo-SP}) and with (\textit{w-SP}) the use of super-patterns.  We next show the
average and maximum runtime for the normalization algorithm.  Finally, we show
the average similarity among the mutated benchmarks, again comparing \textit{pre},
\textit{wo-SP}, and \textit{w-SP}.
\Cref{fig:unique_outputs_smtlib} shows histograms of the number of
distinct benchmarks among the 10 mutations before and after (with super-patterns) normalization.
\begin{figure*}[t]
    \centering
    \subfloat[Before Normalization]{\includegraphics[width=0.45\textwidth]{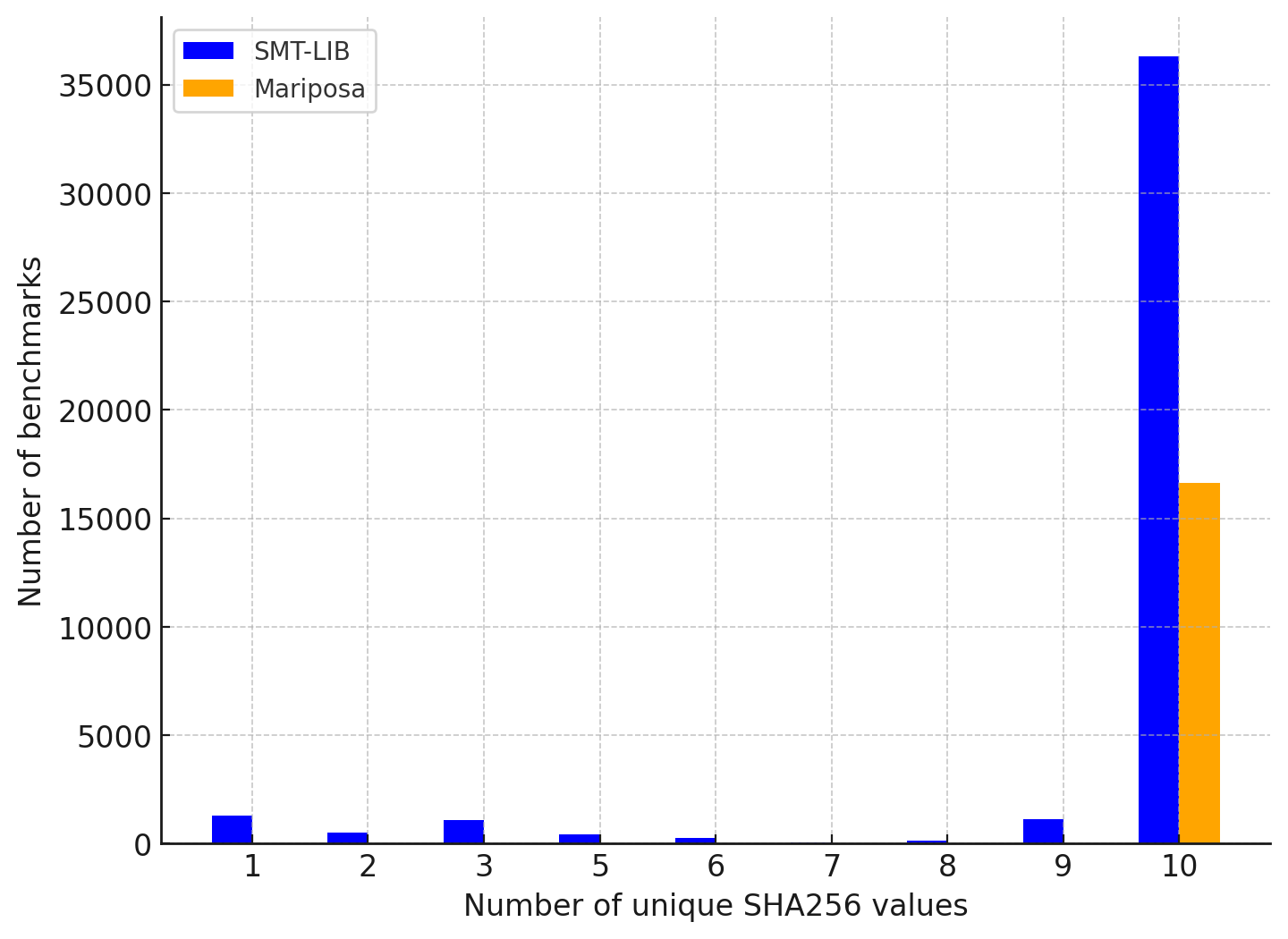}
    \label{fig:unique_outputs_smtlib_scramble}}
    \subfloat[After Normalization]{\includegraphics[width=0.45\textwidth]{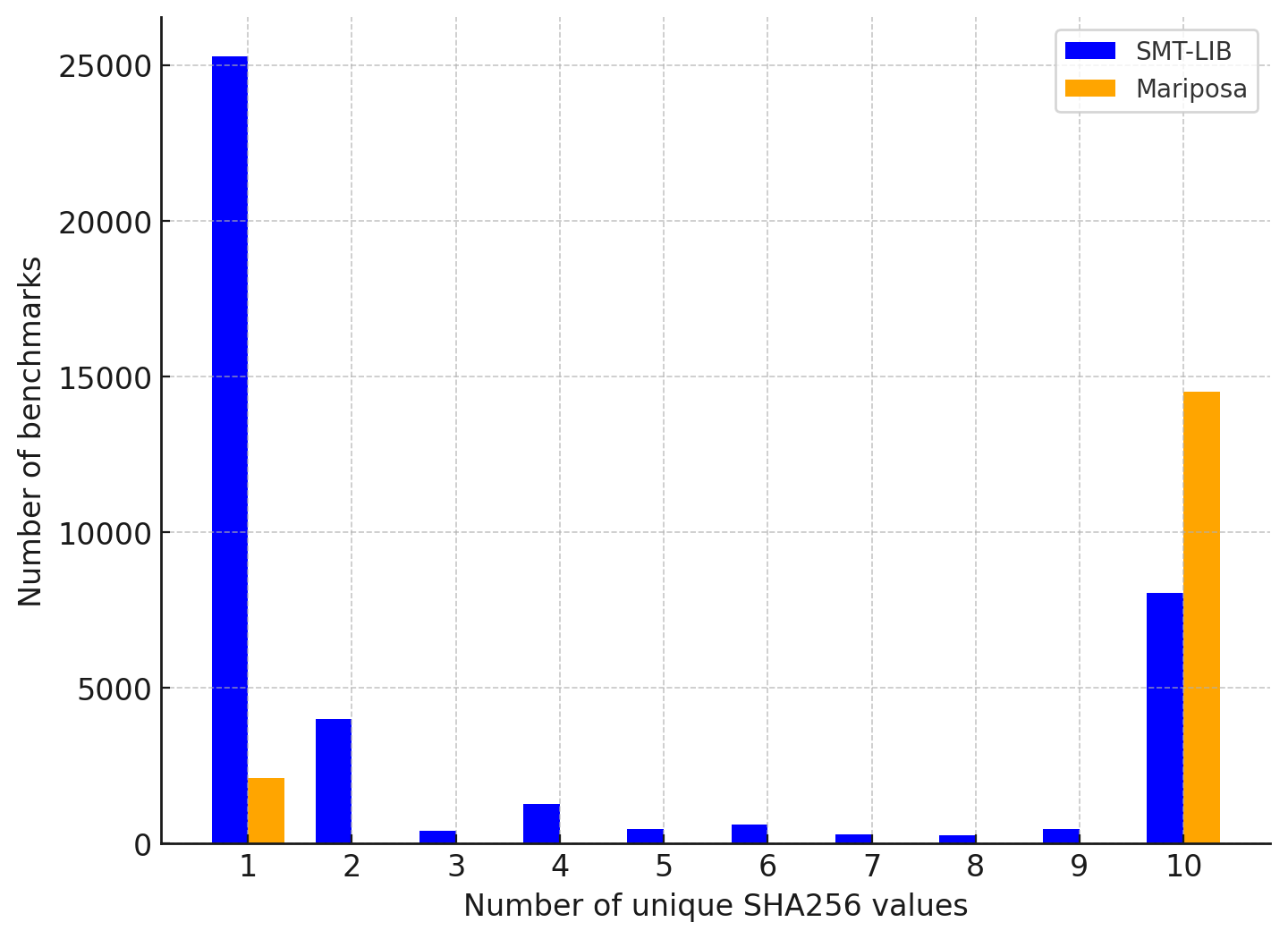}
      \label{fig:unique_outputs_smtlib_norm}}
    \caption{Number of distinct benchmarks before and after normalization.}
    \label{fig:unique_outputs_smtlib}
\end{figure*}

\begin{table*}[t!]
  \centering
  \scalebox{0.92}{%

\begin{tabular}{lr@{\hskip .5em}|@{\hskip .5em}r@{\hskip .5em}|@{\hskip .5em}rrr@{\hskip .5em}|@{\hskip .5em}rr@{\hskip .5em}|@{\hskip .5em}rrr}
  \toprule
  
  \textbf{Division (\benchsmtlib)} & (Excluded/Total)
  & \textbf{\#Bench}
  & \multicolumn{3}{c|@{\hskip .5em}}{\textbf{\#Unique}}
  & \multicolumn{2}{c|@{\hskip .5em}}{\textbf{Time (s)}}
  & \multicolumn{3}{c}{\textbf{Similarity (\%)}} \\
  
  &
  &
  & \textit{pre}
  & \textit{wo-SP}
  & \textit{w-SP}
  & \textit{Avg}
  & \textit{Max}
  & \textit{pre}
  & \textit{wo-SP}
  & \textit{w-SP} \\
  \midrule
  
  Arith                             &      (0/883)  &      883  &      9.33  &     1.05  &     1.02 &   0.0010  &    0.1720  &  7.07  &  99.39  &  99.62 \\
  BitVec                            &     (5/793)  &      788  &      9.36  &     1.07  &     1.03 &   0.0139  &    3.5230  &  5.81  &  99.01  &  99.48 \\
  Equality                          &     (0/1,570)  &     1,570  &     10.00  &     8.48  &     2.79 &   0.0075  &    0.1020  &  0.05  &  47.59  &  98.48 \\
  Equality+LinearArith              &     (0/2,595)  &     2,595  &      9.99  &     7.34  &     5.35 &   0.0030  &    0.0860  &  1.29  &  65.56  &  86.78 \\
  Equality+MachineArith             &     (2/905)  &      903  &     10.00  &     1.25  &     1.22 &   0.0295  &    3.6140  &  0.76  &  97.97  &  98.18 \\
  Equality+NonLinearArith           &    (0/1,452)  &     1,452  &      9.70  &     2.98  &     1.64 &   0.0080  &    0.4980  &  3.14  &  88.89  &  98.48 \\
  FPArith                           &      (0/319)  &      319  &      9.48  &     1.03  &     1.03 &   0.0000  &    0.0010  &  5.25  &  99.61  &  99.61 \\
  QF\_Bitvec                        &   (42/3,320)  &      3,278 &      9.55  &     4.94  &     4.78 &   0.7248  &   38.7570  &  4.63  &  74.17  &  79.88 \\
  QF\_Equality                      &      (0/667)  &       667 &     10.00  &     2.63  &     2.16 &   0.0306  &    1.0710  &  0.73  &  85.67  &  93.24 \\
  QF\_Equality+Bitvec               &   (20/1,208)  &      1,188 &      9.87  &     4.49  &     4.44 &   0.1222  &   13.3070  &  2.94  &  71.07  &  74.06 \\
  QF\_Equality+LinearArith          &     (12/812)  &       800 &      9.95  &     3.66  &     3.38 &   0.0590  &    6.6460  &  1.58  &  83.17  &  89.38 \\
  QF\_Equality+NonLinearArith       &     (0/589)  &       589 &     10.00  &     6.73  &     6.22 &   0.2908  &   29.6690  &  1.76  &  71.40  &  85.08 \\
  QF\_FPArith                       &     (0/4,813)  &      4,813 &      9.08  &     2.80  &     2.52 &   0.0002  &    0.0290  &  13.17 &  90.93  &  92.30 \\
  QF\_LinearIntArith                & (6/10,801)  &      10,795 &      9.98  &     6.35  &     4.21 &   0.2728  &   37.2840  &  8.62  &  78.79  &  88.67 \\
  QF\_LinearRealArith               &    (3/1,466)  &      1,463 &      9.95  &     3.72  &     3.26 &   0.2230  &   30.3610  &  6.77  &  91.19  &  93.86 \\
  QF\_NonLinearIntArith             &     (0/653)  &       653 &      9.88  &     3.04  &     2.64 &   0.0752  &    4.5400  &  4.77  &  94.00  &  95.69 \\
  QF\_NonLinearRealArith            &    (1/3,343)  &      3,342 &      9.83  &     1.83  &     1.67 &   0.0620  &   10.5110  &  8.42  &  96.05  &  97.76 \\
  QF\_Strings                       &    (0/4,977)  &      4,977 &      6.21  &     2.69  &     2.45 &   0.0004  &    0.0190  &  44.50 &  91.75  &  94.93 \\
  
  \midrule
  
    \multicolumn{2}{l}{\textbf{Family (\benchmariposa)}}\\
  
  \midrule
  
    \benchdice                       &          (0/1536) &     1536    &  10.00 &     10.00 &     10.00 &    5.4847 &    12.4670  &  0.93  &  99.45  &  98.75 \\
    \benchkomodod                    &          (0/2054) &     2054    &  10.00 &     10.00 &     10.00 &    1.2750 &    3.7500   &  0.18  &  34.87  &  88.28 \\
    \benchkomodos                    &          (2/773) &     771    &  10.00 &     1.00 &     1.00 &    0.0004 &    0.0790     &  1.09  &  100.00 &  100.00 \\
    \benchverid                      &          (0/5170) &     5170    &  10.00 &     10.00 &     10.00 &    0.0790 &    9.1740   &  0.19  &  18.15  &  87.30 \\
    \benchveril                      &          (0/5334) &     5334    &  10.00 &     10.00 &     10.00 &    1.2533 &    9.6420   &  0.17  &  12.73  &  81.85 \\
    \benchwasm                       &          (0/1755) &     1755    &  10.00 &     3.16 &     3.15 &   0.0657  &    4.2050    &  0.01  &  99.53  &  99.52 \\
  
  \bottomrule
  \end{tabular}

  } 
  \vspace{1ex}
  \caption{Number of unique outputs, normalization time, and similarity of normalized outputs.}
  \vspace{-2ex}
  \label{tab:unique}
\end{table*}

First, we observe that before normalization, we nearly always have
10 distinct versions of a benchmark.  Exceptions (e.g., 
in the QF\_Strings category) occur when there are not a sufficient number of
user-defined symbols and assertions to create 10 distinct versions.  
For the \benchsmtlib benchmark set,
normalization always significantly reduces the number of distinct
versions of a benchmark, and results with super-patterns always improve
over those without super-patterns---sometimes significantly (e.g.,
in the Equality division).

For other benchmarks, however, especially the QF\_Equality+NonLinearArith
division and several families in the \benchmariposa benchmark set, our
normalizer struggles to produce a small number of distinct outputs.  
Investigating these benchmarks
reveals large numbers of assertions with lots of symmetry.
In these cases, the super-pattern comparison fails to distinguish between
assertions with identical patterns, resulting in different normal forms because
of assertions appearing in a non-deterministic order.
It is worth highlighting that for these benchmarks, our normalization algorithm
still achieves a high level of similarity.  As we discuss below, there is also an
improvement in stability for these benchmarks.  This suggests that full
normalization is not necessary to improve stability.  Improved
similarity is already helpful.
%

As mentioned above, over all of the benchmarks, only 87 are excluded due to
timeouts during normalization, and these are typically very large benchmarks
where parsing alone takes most or all of the time.  For the vast majority of the
remaining benchmarks, the normalization overhead is
extremely low (a fraction of a second on average as shown in
\Cref{tab:unique}).  The aggregate overhead, computed as the sum of the
normalizing time for all benchmarks divided by sum of the normalization plus
solving time for all benchmarks, is less than $0.8\%$.  

We also conducted exploratory experiments on a random subset of
\benchsmtlib (4,461 benchmarks) to evaluate the performance of the normalizer
in~\cite{weber16scrambling} in terms of uniqueness of the normalized output.
We observed that it was able to produce unique outputs for only 12\%, of the
benchmarks, which is vastly outperformed by our approach (87\% on this set).
This result, however, is not unexpected since the algorithm
in~\cite{weber16scrambling} was designed for the purpose of exploiting a
specific weakness in the SMT-COMP's scrambling algorithm in use at the time rather than with
general-purpose normalization in mind.

\subsection{Stability}\label{sec:solver-stability}

In our second experiment, we evaluate how our normalization algorithm affects
solver stability.
We use two SMT solvers: cvc5~\cite{cvc5} and Z3~\cite{z3}. These are natural
choices, as both are used extensively and support a wide range of theories.  We
limit our evaluation to these two solvers, as they are the only solvers that
support all of our benchmarks.
For the \benchmariposa benchmark set, we limit our evaluation to Z3.  This is
because these benchmarks come from a specific use case targeting Z3 (as
observed already by Zhou et al.~\cite{Zhou2023}), and most of them
are unsolvable by cvc5.


%
We use penalized runtime (PR-2) and Median Absolute Deviation (MAD) as metrics.
Penalized runtime the sum of the time taken for solved
benchmarks
plus a penalty equal to two times the timeout for each unsolved benchmark
(timeouts, memory outs, or other errors).
PR-2 thus combines elements of both total
time and number of solved benchmarks into a single metric.
MAD measures how much variation there is among a set of results, with
lower numbers indicating less variation and higher numbers indicating more
variation.  Applying MAD to PR-2 scores is thus a good proxy for stability.

Results are shown in~\Cref{tab:par2}
for each division (in \benchsmtlib) and family (in
\benchmariposa).  We report results on the benchmarks before normalization (\emph{no norm.})
and after normalization with super-patterns (\emph{norm.}).
The column labeled \emph{avg} contains the average PR-2 score, computed as
follows.  Recall that we create ten mutations for each benchmark.  Each
mutation is based on a specific random seed.  We compute the total PR-2 score for
all the benchmarks for each seed separately.  We then take the
average of these ten PR-2 scores.
The MAD is also computed over these ten cumulative 
PR-2 scores (one for each seed). 

\begin{table}[t!]
  \centering
  \scalebox{0.795}{\newcommand{\be}[1]{\textbf{#1}}
\begin{tabular}{l@{\hspace{.35em}}|@{\hspace{.35em}}r@{\hspace{.35em}}r@{\hspace{.35em}}r@{\hspace{.35em}}r@{\hspace{.35em}}|@{\hspace{.35em}}r@{\hspace{.35em}}r@{\hspace{.35em}}r@{\hspace{.35em}}r}
  \toprule
  \textbf{Division (\benchsmtlib)}
& \multicolumn{4}{c}{\textbf{cvc5 PR-2}}
& \multicolumn{4}{c}{\textbf{Z3 PR-2}}
\\

  & \multicolumn{2}{c}{\textit{no norm.}}
  & \multicolumn{2}{c|@{\hspace{.7em}}}{\textit{norm.}}
  & \multicolumn{2}{c}{\textit{no norm.}}
  & \multicolumn{2}{c}{\textit{norm}}
\\

  & \textit{avg.}
  & \textit{MAD}
  & \textit{avg.}
  & \textit{MAD}
  & \textit{avg.}
  & \textit{MAD}
  & \textit{avg.}
  & \textit{MAD}
  \\

    \midrule
Arith                       &    17,161 &     93.5 &  17,264 &   \be{2.9} &   21,532 &     120.6 &  11,346  &   \be{0.3} \\
BitVec                      &    29,141 &     41.0 &  29,247 &  \be{14.6} &   34,245 &     121.4 &  28,248  &  \be{14.6} \\
Equality                    &   101,912 &    286.0 & 101,778 &   \be{4.0} &  107,718 &     108.8 & 107,885  &  \be{10.7} \\
Equality+LinearArith        &    77,933 &    208.4 &  77,620 &  \be{32.8} &   72,603 &     254.5 &  73,041  &  \be{22.2} \\
Equality+MachineArith       &    91,935 &     23.4 &  91,518 &   \be{6.8} &   74,037 &     452.8 &  73,539  &  \be{37.2} \\
Equality+NonLinearArith     &    97,462 &    351.0 &  95,961 &  \be{10.8} &   89,615 &     375.5 &  90,551  &  \be{17.1} \\
FPArith                     &    18,089 &     49.6 &  17,920 &  \be{49.5} &   14,371 &     143.2 &  14,214  &  \be{28.5} \\
QF\_Bitvec                  &   149,591 &    282.8 & 153,082 & \be{208.9} &  104,282 &     751.8 &  94,428  & \be{174.6} \\
QF\_Equality                &     2,408 &     49.0 &   2,435 &   \be{3.7} &    1,285 &      64.4 &   1,195  &   \be{1.0} \\
QF\_Equality+Bitvec         &    40,712 &    157.1 &  40,246 & \be{153.7} &   36,282 &     265.5 &  35,870  &  \be{84.7} \\
QF\_Equality+LinearArith    &    14,927 &     46.4 &  16,070 &  \be{44.9} &    8,278 &     296.4 &   4,862  & \be{179.5} \\
QF\_Equality+NonLinearArith &    34,344 &    422.9 &  34,234 & \be{276.0} &   25,811 &     374.6 &  25,731  & \be{175.3} \\
QF\_FPArith                 &    28,971 &    289.6 &  28,063 & \be{108.1} &   56,349 &     324.4 &  55,262  & \be{106.3} \\
QF\_LinearIntArith          &   242,429 &    402.8 & 244,127 & \be{148.7} &  151,935 &     606.8 & 132,783  & \be{345.7} \\
QF\_LinearRealArith         &    30,567 &    313.6 &  29,539 & \be{101.0} &   23,547 &     310.8 &  21,260  & \be{160.4} \\
QF\_NonLinearIntArith       &    34,201 &    615.7 &  34,267 &  \be{26.3} &   23,914 &     194.7 &  22,265  &   \be{6.8} \\
QF\_NonLinearRealArith      &    39,951 &     86.8 &  39,377 &  \be{27.9} &   26,186 &     188.3 &  28,016  & \be{131.7} \\
QF\_Strings                 &    34,092 &    341.4 &  33,255 & \be{121.6} &   56,970 &     321.4 &  56,420  & \be{158.4} \\

\midrule

  \multicolumn{2}{l}{\textbf{Family (\benchmariposa)}}\\

\midrule
  \benchdice               &       --  &     --   &     --  &        --  &     7,531 &      559.3 &   8,951 & \be{545.3} \\
  \benchkomodod            &       --  &     --   &     --  &        --  &    12,561 &      264.2 &  13,327 &  \be{59.9} \\
  \benchkomodos            &       --  &     --   &     --  &        --  &     3,284 &       79.9 &   3,181 &  \be{17.2} \\
  \benchverid              &       --  &     --   &     --  &        --  &    12,623 &      414.4 &  12,612 & \be{169.0} \\
  \benchveril              &       --  &     --   &     --  &        --  &    24,889 &      203.6 &  23,971 & \be{172.7} \\
  \benchwasm               &       --  &     --   &     --  &        --  &     2,919 &        9.5 &   3,018 &   \be{0.2} \\
\bottomrule
\end{tabular}
 }
  \vspace{1ex}
  \caption{
  PR-2 and MAD scores on mutated benchmarks before and after normalization.
}
\vspace*{-2em}
\label{tab:par2}
\end{table}

We see that, for both solvers,
the performance (avg. column) on normalized benchmarks is generally comparable with
that of non-normalized benchmarks, showing
that, on average, normalization does not appear to greatly affect runtime.

More importantly, the MAD score
improves significantly in \emph{all} cases, sometimes by more than an order
of magnitude (e.g., in the Equality division).
This strongly suggests that our normalization algorithm improves stability.
Even for benchmarks for which the normalizer completely
fails to produce normal forms (e.g., the \benchmariposa benchmarks that always
still have 10 distinct benchmarks after normalization), the stability improves,
sometimes significantly.  This provides promising evidence that full
normalization is not required for improving stability.  An improvement in
similarity may be sufficient.


\subsection{Code Optimizations} \label{sec:optimization}
Some of the benchmarks in our benchmark sets are extremely large, with sizes up
to hundreds of megabytes and with up to hundreds of thousands of 
user-defined symbols and assertions.
In this section, we discuss three optimizations that are crucial for
scalability of our algorithm on such benchmarks.

\paragraph{Pattern Compression}
As presented in~\Cref{sec:formal}, formulas are sequences whose
size corresponds to the size of their abstract syntax tree.
However, internally, SMT solvers represent formulas as directed acyclic graphs
(DAGs), because they often share subterms.  The DAG representation can be
exponentially more concise.  It is thus essential to also
incorporate this kind of sharing in our representation of patterns.
In our implementation, we represent a pattern as an index into a dictionary of
the subterms
in the tree representation of the pattern.  For example, the formula $\langle +,*,x,x,*,x,x\rangle$
would be represented as the index $\textbf{2}$ into the dictionary $(\textbf{1}:\langle *,x,x\rangle, \textbf{2}: \langle +,\textbf{1},\textbf{1})$.


\paragraph{Super-pattern Computation}
A naive way to compute the super-pattern for a symbol $s$ is to traverse all 
of the assertions and look up the role of $s$
in every assertion. However, this approach is not scalable for large benchmarks.
To address this issue, we perform a one-time indexing pass over the assertions
to create arrays for each symbol $s$ containing
pointers to only the assertions in which $s$ appears. This way, we only
need to traverse these assertions when calculating the super-pattern for $s$.
This optimization reduces the number of lookups from hundreds of millions to
hundreds of thousands on some problematic benchmarks, improving the
normalization time by more than an order of magnitude.

\paragraph{Super-pattern Compression}
This optimization leverages the sparsity of super-patterns: in large
benchmarks, each symbol  typically appears in only a small subset of the assertions. 
Thus, the role of a symbol is likely to be 0 in most assertions. 
To exploit this, we represent super-patterns as sequences consisting of
the non-zero role values interleaved with counts of the number of zeros.


\section{Conclusion}

Our normalization algorithm is a promising step towards a more stable and
predictable SMT solving experience.  It generally scales well and is applicable to a wide range of benchmarks and logics.
We have shown that it can often produce a single unique normal form
for a set of mutated benchmarks, and even when it cannot, it greatly increases
the similarity of the benchmarks.  We also saw that, on average,
it improves stability without significantly affecting performance.

In future work, we intend to expand our normalizer to handle other mutations such
as antisymmetric operator replacement and operand swapping for commutative
operators.  We also plan to explore whether additional normalization techniques
can be used to further improve our results.
Finally, we plan to investigate the use of our normalization algorithm
as a preprocessing step to improve the hit rate of applications that cache
formulas in order to reuse solving results.

\label{sec:conclusion}

\bibliographystyle{plain}
\bibliography{bibliography}

\pagebreak
\section{Appendix}
\label{sec:appendix}

\begin{lemma}\label{lem:csp}
    Let $\seqassertion$ be a sequence of formulas.  If $\seqassertion_R = R(\seqassertion_P)$, where $\seqassertion_P\in \pat{C(S(\seqassertion))}$ and $R$ is
    a renaming, then $\pat{C(S(\seqassertion_R))} = \pat{C(S(\seqassertion))}$.
\end{lemma}

\begin{proof}

    First, note that for any formula sequence $\seqassertion$, we have
    $R(\pat{C(S(\seqassertion))})=C(S(R(\pat{\seqassertion})))$.  This is because
    the first generates a set of equivalent formula sequences and then renames the symbols,
    while the second renames the symbols and then generates a set of equivalent
    formula sequences, but these two transformations are indpendent, so the
    result is the same, regardless of which order they are done in.

    Now, we can write $\seqassertion_R=R(\pat{\seqassertion'})$ for some
    $\seqassertion'\in C(S(\seqassertion))$.  It follows that
    $\pat{C(S(\seqassertion_R))}=\pat{C(S(R(\pat{\seqassertion'})))}=\pat{R(\pat{C(S(\seqassertion'))})}=\pat{C(S(\seqassertion'))}$,
    since computing the pattern of a renaming of a pattern is just the same as
    computing the pattern.  It remains to show that $\pat{C(S(\seqassertion'))}=\pat{C(S(\seqassertion))}$.

    But $\seqassertion'\in C(S(\seqassertion))$, which
    means that $\seqassertion$ can be obtained from $\seqassertion'$ by
    swapping zero or more commutative operators and permuting the order of the
    formulas in $\seqassertion'$.  But then, any element of $C(S(\seqassertion))$ can be obtained
    from $\seqassertion'$ by swapping commutative operators and permuting the
    order, so $C(S(\seqassertion'))$ is just the same as $C(S(\seqassertion))$.
    Thus, $\pat{C(S(\seqassertion'))}=\pat{C(S(\seqassertion))}$.
\end{proof}

\begin{lemma}\label{lem:min}
    For any formula $\seqassertion$, $\pat{C(S(\seqassertion)))}$ has a unique minimal element, according to
    $\prec$.
\end{lemma}
\begin{proof}
    When comparing two elements of $\pat{C(S(\seqassertion))}$, each of which is a sequence of
    patterns, we concatenate the patterns together and compare them with
    $\prec$.  Since $\prec$ is a total order, one of them is always smaller.
    Thus there is a minimal element of $\pat{C(S(\seqassertion))}$.
\end{proof}


\setcounter{section}{3}
\setcounter{theorem}{7}
\begin{theorem}
    Function $\N$ is normalizing.
  \end{theorem}
  \begin{proof}
    Notice that $\N(\seqassertion) = \R(\seqassertion')$, where $\R$ is the
    identity function and $\seqassertion'$ is the $\prec$-minimal element of
    $\pat{C(S(\seqassertion))}$.  It is not hard to see that $\R$ is also a
    renaming, so the first requirement is met.  Now, let
    $M_1= R_1(M'_1)$ and $M_2= R_2(M'_2)$, with
    $M'_1,M'_2 \in \pat{C(S(\seqassertion))}$, where $R_1,R_2$ are renamings.
    By Lemma~\ref{lem:csp}, $\pat{C(S(M_1))} = \pat{C(S(\seqassertion))} = \pat{C(S(M_2))}$.  But
    $\pat{C(S(\seqassertion))}$ has a unique minimal element, $\seqassertion'$
    by Lemma~\ref{lem:min}.
    Thus, $\N(M_1) = \seqassertion' = \N(M_2)$.
  \end{proof}

  \setcounter{section}{3}
  \setcounter{theorem}{8}
  \begin{theorem}
    Let N be a normalizing function.  Then, computing $N(\seqassertion)$ for an arbitrary
    $\seqassertion$ is as hard as solving graph isomorphism.
  \end{theorem}
  
  \begin{proof}
  Let $G_1=(V_1,E_1)$ and $G_2=(V_2,E_2)$ be two undirected graphs.  Assume
  without loss of generality that $V_1 \cap V_2 = \emptyset$.  For each
  $v\in\{V_1\cup V_2\}$, let $u(v)$ map $v$ to some unique user-defined symbol
  (i.e., $u$ is injective). Now, define
  $A_i = \{ \langle f,x\rangle \mid x\in V_i\} \cup \{ \langle =,x,y \rangle \mid \exists\,(v,v')\in E_i.\:\{x,y\} = \{u(v),u(v')\}\}$,
  where $f$ is some Boolean predicate.  Note that $G_i$ can be recovered from $A_i$,
  simply by creating a vertex for every user-defined symbol appearing as an
  argument of $f$ in $A_i$
  and then adding an edge between $u^{-1}(x)$ and $u^{-1}(y)$ whenever the
  formula $\langle
  =,x,y\rangle$ appears in $A_i$.  Furthermore, shuffling the formulas in $A_i$,
  permuting the order of the operands in equalities (the only commutative
  operator appearing in $A_i$), or renaming the user-defined symbols does not change the
  structure of the graph being represented.  Thus, all elements of
  $\pat{C(S(A_i))}$ represent isomorphic graphs.
  
  Now, suppose $N(A_1) = N(A_2)$,
  For $i\in[1,2]$, we know that $N(A_i) = R(A_i')$ for some $A_i' \in \pat{C(S(A_i))}$.
  Furthermore, because renamings are injective, $A'_1 = \pat{R(A_1}) = \pat{N(A_1}) =
  \pat{N(A_2}) = \pat{R(A_2}) = A'_2$.  But for $i\in[1,2]$, the graph represented by $A_i$ is
  isomorphic to the graph represented by $A'_i$.  Thus the graph represented by
  $A_1$, namely $G_1$, is isomorphic to the graph represented by $A_2$, which is
  $G_2$.
  
  On the other hand, suppose $G_1$ and $G_2$ are isomorphic.  Let $h:V_1\to V_2$
  be the isomorphism function for the graph vertices.  Let $R$ be a renaming such
  that $\forall\,v.\:R(\pat{v}) = h(v)$.  Applying this renaming to
  $A_1$ must be equivalent (modulo order) to $A_2$.  In other words, $A_1 \in
  R(\pat{S(A_2)})$.  Then, because $N$ is normalizing, we must have $N(A_1)=N(A_2)$.
  Thus, computing $N$ is at least as difficult as graph isomorphism.
  \end{proof}

\end{document}